\documentclass[abstract=true]{scrartcl}
\usepackage{authblk}

\usepackage{hyperref}
\usepackage{natbib}
\def\U{{\vec{U}}}
\def\u{{\vec{u}}}

\usepackage{multirow}
\usepackage{multicol}
\usepackage{array,booktabs}
\usepackage{algorithm, algorithmic}

\usepackage{graphicx}
\graphicspath{{figures/}{../figures/}}

\usepackage{adjustbox}

\usepackage{amsfonts, amsmath, amssymb, amsthm, bm, 
            latexsym, mathrsfs, thmtools, wasysym}
\usepackage{cleveref}
\declaretheorem[name=Theorem, refname={theorem, theorems}, Refname={Theorem, Theorems}]{theorem}
\declaretheorem[name=Definition, refname={definition, definitions}, Refname={Definition, Definitions}, sibling=theorem]{definition}
\declaretheorem[name=Lemma, refname={lemma, lemmas}, Refname={Lemma, Lemmas}, sibling=theorem]{lemma}
\declaretheorem[name=Assumption, refname={assumption, assumptions}, Refname={Assumption, Assumptions}, sibling=theorem]{assumption}
\declaretheorem[name=Notation, refname={notation, notations}, Refname={Notation, Notations}, sibling=theorem]{notation}

\renewcommand{\vec}[1]{\mathbf{#1}}

\usepackage{color}

\definecolor{darkgreen}{rgb}{0.3, 0.5, 0.0}

\def\flexts{{FlexCodeTS}}
\def\hbu{{\widehat{\beta}_i(\u)}}
\def\bu{{\beta_i(\u)}}
\def\hfi{{\hat f_I(y|\u)}}

\def\E{{\mathbb E}}
\def\I{{\mathbb I}}
\def\P{{\mathbb P}}
\def\U{{\mathbf{U}}}
\def\u{{\mathbf{u}}}
\def\V{{\mathbb V}}

\def\x{{\vec{x}}}

\def\sF{{\mathcal{F}}}
\def\sY{{\mathcal{Y}}}

\usepackage{longtable}

\title{Flexible conditional density estimation for time series}
\author[1]{Gustavo Grivol\thanks{ggrivol@gmail.com}}
\author[2]{Rafael Izbicki\thanks{rafaelizbicki@gmail.com}}
\author[3]{Alex A. Okuno\thanks{akira.okuno@outlook.com}}
\author[3]{Rafael B.Stern\thanks{rbstern@gmail.com}}
\affil[1]{Stern School of Business, New York University, USA}
\affil[2]{Departamento de Estatística, Universidade Federal de São Carlos, Centro de Ciências Exatas e de Tecnologia, Brazil}
\affil[3]{Departamento de Estatística, Instituto de Matemática e Estatística, Universidade de São Paulo, Brazil}
\begin{document}

\maketitle

\begin{abstract}
 This paper introduces \flexts,
 a new conditional density estimator for
 time series. \flexts \ is a flexible 
 nonparametric method,
 which can be based on an arbitrary
 regression method.
 It is shown that \flexts \ inherits
 the rate of convergence of 
 the chosen regression method.
 Hence, \flexts \ can adapt its
 convergence by employing 
 the regression method that 
 best fits the structure of data.
 From an empirical perspective,
 \flexts \ is compared to
 NNKCDE and GARCH in both 
 simulated and real data.
 \flexts\ is shown to generally obtain
 the best performance among
 the selected methods according to
 either the CDE loss or the pinball loss.
 \par\vskip\baselineskip\noindent
 \textbf{Keywords:} 
 nonparametric statistics,
 conditional density estimation,
 time series,
 machine learning \\ 
 \textbf{MSC:} 00-01; 99-00
\end{abstract}

\section{Introduction}

Predicting future values of 
a time series is often not enough.
For instance, traders can make better option pricing and 
estimate value at risk by 
forecasting the {\em full} uncertainty about 
the future value of an asset, $Y_{t+1}$.  
This evaluation can be made by estimating the 
conditional density, $f(y_{t+1}|y_1,\ldots,y_t,\x_t)$, 
where $\x_t$ are exogenous variables. 
An estimate of $f$ allows one to
construct predictive intervals 
for new observations \citep{Fernandez-Soto,izbicki2020flexible,izbicki2022cd,dey2022calibrated}, 
perform simultaneous quantile regression \citep{Takeuchi}, and 
calculate arbitrary moments of the series, 
such as volatility, skewness and kurtosis 
\citep{filipovic2012conditional,kalda2013nonparametric,GneitingKatzfuss}.
Also, empirical models for asset pricing are incomplete 
unless the full conditional model is specified 
\citep{hansen1994autoregressive}.

In order to obtain conditional density estimators (CDEs),
one often uses parametric models, 
such as ARMA-GARCH \citep{engle1982autoregressive,hamilton2020time}
and generalizations \citep{bollerslev1987conditionally,hansen1994autoregressive,sentana1995quadratic}.
However, the assumptions in these models are
often too restrictive, specially when
the distribution of $Y_{t+1}$ is asymmetric or multimodal.
It can also be hard to capture the relationship between
$Y_{t+1}$ and exogenous variables through a parametric model.
In such cases, nonparametric estimators are useful \citep{hansen1994autoregressive,jeon2012using}. 

The majority of the proposed nonparametric CDEs for time series are
based on kernel smoothers 
\citep{jeon2012using,hardle1997review,hyndman2002nonparametric,fan2008nonlinear}.
This approach performs poorly when
either many lags or exogenous variables do not affect the response,
unless a different bandwidth is used for each covariate \citep{hall2004cross}.
Since selecting a large number of bandwidths is time consuming
\citep{izbicki2016nonparametric,izbicki2017photo}, 
kernel smoothers are effective only in low dimensions \citep{izbicki2014high}. 
Thus, nonparametric CDEs that scale better 
with the number of covariates are needed.

This paper introduces \flexts, 
a version of FlexCode \citep{IzbickiLeeFlexCode} 
tailored for time series.
\flexts \ transforms \emph{any} regression method 
into a time series CDE. Furthermore,
it often inherits the properties of the chosen regression method. 
For instance, \flexts \ based on 
random forests \citep{breiman2001random} yields 
a CDE that perform automatic variable selection and
works well in high dimensional settings. 
Also, \flexts \ based on XGBoost \citep{chen2015xgboost}
scales for large datasets. 

\Cref{sec:method} defines \flexts \ and
describes its implementation.
\Cref{sec:theory} derives 
a general rate of convergence for \flexts \ 
as a function of the rate of convergence of
the chosen regression method.
It also derives a specific rate of convergence
for \flexts \ when 
Nadaraya-Watson regression is used.
\Cref{sec:experiment} compares 
the performance of \flexts \ 
to NNKCDE and GARCH on 
both simulated and real data.

%We  show that FlexCodeTS can indeed give better results for... (dependendo dos resultados, falar de quantile regression tb).

\section{Methodology}
\label{sec:method}

Let the data be $(\u_1,y_1),\ldots,(\u_T,y_T)$, 
where $y_t \in \sY \subseteq \mathbb{R}$ is 
of interest and $\u_t \in \mathbb{R}^d$ contains lagged values of $y_t$ 
as well as available exogenous variables. 
We assume that this series is stationary and, therefore,
the conditional density of interest, 
$f_{y_t|\u_t}(y|\u)$, does not depend on $t$.
Hence, we refer to this density by $f(y|\u)$.

In order to develop \flexts,
the first step is to decompose
$f(y|\u)$ into an orthonormal series.
As long as $\int f^2(y|\u)dy<\infty$, one obtains
for each orthonormal basis over $\mathcal{L}^2(\mathbb{R})$,
$(\phi_{i})_{i \geq 0}$,
the decomposition:
\begin{definition}[Decomposition of $f(y|\u)$]
 \label{def:f_decomp}
 \begin{align*}
 f(y|\u) &= \sum_{i \geq 0}\beta_i(\u)\phi_{i}(y),
 & \text{where }
 \beta_i(\u) = \int f(y|\u)\phi_i(\u)dy
\end{align*}
\end{definition}
Choices for $\phi$ include, 
for instance, the Fourier basis
and wavelets, which are 
more appropriate when $f(y|\u)$ is
spatially inhomogeneous in $y$ \citep{mallat1999wavelet}.
 
Furthermore, if $f$ is smooth 
with respect to $(\phi_{i})_{i \geq 0}$, then 
it is well approximated by the initial terms of the series:
\begin{align}
 \label{eq:cde-1}
 f(y|\u) 
 &\approx f_I(y|\u) 
 := \sum_{i = 0}^I\beta_i(\u)\phi_{i}(y),
 & \text{for some $I \in \mathbb{N}$.}
\end{align}
\flexts \ obtains an estimate for $f(y|\u)$ by 
choosing a value of $I$,
estimating $\beta_1(\u)),\ldots,\beta_I(\u)$, and
plugging in these estimates into \cref{eq:cde-1}.
Hence, \flexts \ breaks CDE
into two tasks: estimating $\beta_i(\u)$ and
choosing $I$.

In order to estimate $\beta_i(\u)$ an
arbitrary regression method can be used.
The key insight is that 
$\beta_i(\u)$ is the projection 
of $f$ over $\phi_i$. Therefore, 
$\beta_i(\u) = \E\left[\phi_i(Y)|\u\right]$.
That is, $\beta_i(\u)$ is the regression of $\phi_i(Y)$ onto $\u$.
Hence, an estimate of $\beta_i(\u)$,
$\hbu$, can be obtained by
applying an arbitrary regression method to $(\u_1,\phi_i(y_1)),\ldots,(\u_T,\phi_i(y_T))$. 

Next, the value of $I$ controls
the bias-variance trade-off and 
is chosen through temporal data splitting
\cref{fig:data_split}.
A \flexts \ estimator, $\hfi$,
is defined for each choice of $I$:
\begin{definition}[\flexts \ CDE]
 \label{def:flexts}
 \begin{align*} 
  \hfi &= 
  \sum_{i= 0}^I \hbu \phi_{i}(y).
 \end{align*}
\end{definition}

\begin{figure}
 \centering
 \includegraphics[trim={2cm 3cm 6.5cm 3cm},clip,width=0.6\textwidth]{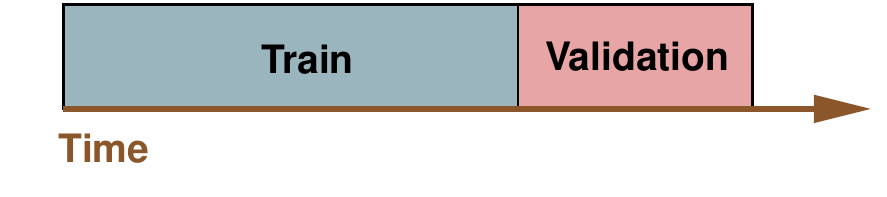}
 \caption{Data-splitting used to tune the cutoff $I$ in FlexCode.} 
 \label{fig:data_split}
\end{figure}

The performance of each
$\hfi$ is evaluated according to
the CDE loss:
\begin{align*}
 L(\hat f_I, f)
 &= \int \left(\hfi -f(y|\u)\right)^2 
 dy d\P(\u).
\end{align*} 

Although the CDE loss depends on
the unknown CDE, $f$,
it can be estimated up to
a constant. First,
the CDE loss admits 
the following simplification:
\begin{align}
 \label{eq:cde-2}
 L(\hat f_I, f)
 &= \int \E[\hat f_I^2(y|\U)]dy -2 \E[\hat f_I(Y|\U)] + K
\end{align}
Using the above equation, 
the CDE loss can be estimated 
up to a constant.
The training set estimates $\beta$ and
the validation set evaluates 
the expression in \cref{eq:cde-2}
by substituting expectations by empirical averages:
\begin{align}
 \label{eq:cde-3}
 \hat{L}(\hat f_I, f)
 &= \frac{\sum_{i=T+1}^{T^*}
 {\int \hat f_I^2(y|\U_i)}dy}{T^*-T} 
 -2\frac{\sum_{i=T+1}^{T^*}
 {\hat f_I(Y_i|\U_i)}}{T^*-T}
\end{align}
Finally, the value of $I$ is chosen by
minimizing $\hat{L}(\hat f_I, f)$.

One of the main advantages of 
\flexts \ is that 
the choice of the regression method can 
follow the structure of the problem at hand. 
For instance, if there is a sparse relationship between
$Y_t$ and $\U_t$, 
XGBoost and random forests 
\citep{chen2015xgboost}, perform 
variable selection.
If $\U_t$ belongs to a low-dimensional submanifold, then 
kNN and spectral series 
\citep{LeeIzbickiReg} achieve optimal rates.

Next, we provide theoretical results
regarding the rate of convergence of
the proposed CDE.

\section{Theory}
\label{sec:theory}

This Section discusses the theoretical properties of \flexts.
Two main results are provided. 
First, under mild conditions,
we determine the rate of convergence
of \flexts \ as a function of
the rate of convergence of 
the estimates of the
regression functions, $\beta$.
As a result, if 
the regression method that is used for
estimating $\beta$ is consistent,
then \flexts \ is consistent.
Second, under additional weak assumptions,
if $\beta$ is estimated using
the Nadaraya-Watson estimator, then
the rate of convergence of 
\flexts \ is obtained. 

For all of the above results,
data is assumed to be stationary:

\begin{assumption}[Stationarity]
 \label{ass:stationary}
 $\{(\U_i,Y_i):i \in \mathbb{N}\}$ is stationary.
\end{assumption}

Also, the target $f(y|\u)$ is
assumed to be smooth over $y$ 
for every $\u$.
Specifically, we assume that
the $s$-th weak derivative of
$f(\cdot|\u)$ is 
uniformly bounded over $\u$.
This concept is formalized by
defining that the functions 
$f(\cdot|\u)$ belong to
similarly smooth Sobolev spaces.
\begin{definition}[\citet{wasserman}]
 For $s>\frac{1}{2}$ and $c > 0$,
 the Sobolev space, 
 $W_{\phi}(s,c)$ is
 \begin{align*}
  W_{\phi}(s,c)
  &= \left\{f \equiv 
  \sum_{i = 1}^{\infty} 
  \beta_i \phi_i :  
  \sum_{i = 1}^{\infty} 
  (\pi i)^{2s} \beta^2_i \leq c^2 
  \right\}
 \end{align*}
\end{definition}
\begin{assumption}[Smoothness in 
 the $y$ direction]
 \label{ass:sobolevZ} 
 $f(\cdot|\u) \in 
 W_{\phi}(s(\u),c(\u))$, 
 where $\inf_\u s(\u) 
 := \gamma > \frac{1}{2}$ and
 $\E[c^2(\U)] := C < \infty$.
\end{assumption}

The next subsection discusses 
the rate of convergence of \flexts \
as a function of 
the rate of convergence of
the estimates of $\beta$.

\subsection{General rate of convergence of \flexts}

In order to study the rate of convergence of
\flexts \ this section assumes that
the functions $\beta_i(\u)$ are 
estimated using regressors,
$\hbu$, which converge uniformly
at a rate $O_P(T^{-r})$.
\begin{assumption} [Regression convergence]
 \label{ass:regression}
 There exists $r > 0$ such that
 \begin{align*}
  \sup_i \int \left(\hbu - \bu \right)^2d\P(\u)
  = O_P(T^{-2r})
 \end{align*}
\end{assumption}

Using \cref{ass:regression}, 
it is possible to obtain a rate of convergence for
\flexts \ under the CDE loss:
\begin{theorem}
 \label{thm::main}
 Under \Cref{ass:stationary,ass:sobolevZ,ass:regression}, \flexts \ 
 (\cref{def:flexts}) satisfies
 the following bound on 
 the CDE loss,
 \begin{align*}
  \iint\left(\widehat{f}_I(y|\u)
  -f(y|\u)\right)^2dyd\P(\u) 
  &\leq IO_P\left(T^{-2r}\right)
  +O\left(I^{-2\gamma}\right).
 \end{align*}
 Hence, by taking the optimal rate 
 $I \sim T^{\frac{2r}{2\gamma+1}},$
 \begin{align*}
  \iint\left(\widehat{f}_I(y|\u)
  -f(y|\u)\right)^2dyd\P(\u)
  &\leq O_P\left(T^{-\frac{4\gamma r}
  {2\gamma+1}}\right).
 \end{align*}
\end{theorem}
The rate of convergence in 
\cref{thm::main} depends on
the rate convergence of $\hbu$, 
as in \cref{ass:regression}.
The rate in \cref{ass:regression} is 
typically related to 
the number of covariates and to 
smoothness assumptions regarding $\beta$.
\Cref{sec:NW} shows that,
under suitable conditions,
the Nadaraya-Watson regression estimator
satisfies \cref{ass:regression} and, therefore,
\cref{thm::main} can be applied.

\subsection{Rate of convergence of \flexts \ using Nadaraya-Watson regression}
\label{sec:NW}

This section studies the convergence of
\flexts \ when $\beta_i$ is estimated using
a Nadaraya-Watson regression with a uniform kernel:
\begin{definition}[Nadaraya-Watson regression with uniform kernel]
 \label{def:nw}
 \begin{align*}
   \widehat{\beta}_i (\u)
   &= \frac{\sum_{j=1}^T 
   \I(\|\U_j - \u\|_2 \leq \delta_T) \phi_i(Y_j)}
   {\sum_{j=1}^T \I(\|\U_j - \u\|_2 \leq \delta_T)},
   & \text{for some choice of } \delta_T.
 \end{align*}
\end{definition}

The rate of convergence for $\hbu$ is
obtained following closely the results in
\cite{Truong1992}. In order to apply these results,
additional assumptions are required.
First, in order to obtain that
every $\bu$ is smooth over $\u$,
\cref{ass:smoothness} ensures that
$f(y|\u)$ is Lipschitz continuous:
\begin{assumption}[Smoothness in  $\u$ direction]
 \label{ass:smoothness}
 $(Y,U)$ is bounded and there exists 
 $M_0 >0$ such that,
 for every $y \in \sY$,
 \begin{align*}
  |f(y|\u_1)-f(y|\u_2)|\leq M_0 ||\u_1-\u_2||.
 \end{align*}
\end{assumption}
Similarly, the value of $\bu$ is bounded.
Since $\bu = \E[\phi_i(Y)|\U]$,
this restriction is obtained by
bounding the image of $\phi_i$.
This assumption is satisfied, for instance,
by the cosine and the trigonometric basis.
\begin{assumption}
 \label{ass:bounded_basis}
 There exists $M > 0$ such that
 $\sup_{i \in \mathbb{N}, y \in \sY}
 |\phi_i(y)| < M$.
\end{assumption}
Also, the joint distribution of
$\U_0,\U_1,\ldots,\U_T$ does not have
extremal high probability or
low probability regions:
\begin{assumption}[Bounded density]
 \label{ass:densityU}
 \label{ass:conditionalDensityU}
 The distribution of $\U_0$ is 
 continuous and its density is 
 bounded away from zero and infinity.
 Also, the conditional distribution of
 $\U_j$ given $\U_0$ is 
 bounded away from  zero and infinity.
\end{assumption}
Finally, $(\U_i,Y_i)$ is assumed to
have a short memory. Such a constraint is
obtained by assuming the series to be
$\alpha$-mixing:
\begin{assumption}[Strong $\alpha$-mixing]
 \label{ass:alphaMixing}
 Let $\mathcal{F}_0$ and 
 $\mathcal{F}^{k}$ denote the 
 $\sigma$-fields generated by 
 $\{(\U_i,Y_i):i\leq 0\}$ and 
 $\{(\U_i,Y_i):i\geq k\}$, respectively.
 Also, let
 \begin{align*}
  \alpha(k) &=
  \sup_{A \in \sF_0, B \in \sF^k} 
  \left|\P(A \cap B)-\P(A)\P(B)\right|.
 \end{align*}
 There exists $\rho \in (0,1)$ such that
 $\alpha(k) = O\left(\rho^k\right)$.
\end{assumption}

Under the above assumptions,
it is possible to obtain
the rate of convergence of \flexts \ 
with Nadaraya-Watson regression:

\begin{theorem}
 \label{thm:NW}
 If $\hbu$ is such as in \cref{def:nw} and
 satisfies $\delta_T \sim T^{-\frac{1}{2+d}}$,
 and $\hfi$ is such as in \cref{def:flexts} with
 $I \sim T^{\frac{2}{(2\gamma+1)(2+d)}}$, then  
 under \cref{ass:stationary,ass:sobolevZ,ass:smoothness,ass:bounded_basis,ass:densityU,ass:alphaMixing},
 \begin{align*}
  \iint \left(\hfi - f(y|\u)\right)^2 dy d\P(\u)  
  &= O_P\left(T^\frac{-2\gamma}
  {2\gamma+1+d \left(\gamma+\frac{1}{2}\right)} \right).
 \end{align*}
\end{theorem}

The same proof strategy used in \cref{thm:NW} 
can be used to derive specific rates 
for other regression estimators for
$\bu$ when the convergence rate is known. 
For example, this is the case for 
the neural networks-based method 
explored by \cite{chen1999improved}.

The next section evaluates 
the empirical performance of 
\flexts \ in real and simulated datasets.

\section{Experiments}
\label{sec:experiment}

This section compares 
\flexts \ to two other CDEs:
NNKCDE \citep{pospisil2019tpospisi,dalmasso2020conditional,izbicki2020nnkcde}
and GARCH \citep{bollerslev1986generalized}.
NNKCDE takes the $k$-nearest neighbors of 
an evaluation point, $\u$, and uses
these points to estimate
$f(y|\u)$ through a kernel density estimator.
GARCH assumes that $Y_{t}$ follows 
an ARMA model and that the residuals
follow a SGARCH model.

While \cref{sec:simulated} compares 
these methods using simulated data, 
\cref{sec:real_data} compares them
using real data regarding
criptocurrency prices and also
the Individual Household Electric 
Power Consumption data from UCI.

\subsection{Simulated Data}
\label{sec:simulated}

This section compares 
\flexts, NNKCDE and GARCH using
simulated data. We generate data using 
$3$ sample sizes: $n = 1000$, $2500$ and $5000$.
The following $6$ simulation scenarios are studied:

\begin{enumerate}
 \item \textbf{AR}: 
 $Y_t = 0.2 Y_{t-1}+ 0.3 Y_{t-2}
 +0.35 Y_{t-3} +\epsilon_t$,
 where $\epsilon_t \sim N(0,1)$.
 
 \item \textbf{ARMA JUMP}:
 $Y_t =  0.1Y_{t-1} + 0.4Y_{t-2} 
 + 0.4Y_{t-3} +0.01
 -0.3Z_t + 0.05(1+Z_t)\epsilon_t$,
 where $\epsilon_t \sim N(0,1)$ and
 $Z_t \sim \text{Bernoulli}(0.05)$.
 
 \item \textbf{ARMA JUMP T}:
 $Y_t =  0.1Y_{t-1} + 0.4Y_{t-2} + 0.4Y_{t-3}
 +0.01-0.3Z_t + 0.05(1+Z_t)\epsilon_t$,
 where $\epsilon_t \sim t_3$ and
 $z_t \sim \text{Bernoulli}(0.05)$,
 where $t_3$ denotes Student's t-distribution with
 $3$ degrees of freedom.
 
 \item \textbf{NONLINEAR MEAN}: 
 $Y_t = \sin^2\left(\pi Y_{t-3}\right)+\epsilon_t$, where $\epsilon_t\sim N(0,\sigma^2)$.
 
 \item \textbf{NONLINEAR VARIANCE}: 
 $Y_t \sim N(0,\sigma_t^2)$, 
 where $\sigma_t=0.1$ if $|Y_{t-3}|>0.5$ and
 $\sigma_t=1$, otherwise.
 
 \item \textbf{JUMP DIFFUSION}: 
 The model proposed by \cite{Christoffersen2016},
 ``Time-varying Crash Risk: The Role of Market Liquidity''.
\end{enumerate}

The code for the above simulations 
as well as for training and 
evaluating the CDEs was implemented in
the R language.
The estimators were evaluated according to
the CDE loss (\cref{eq:cde-3}) and
the quantile loss \citep{koenker2001quantile}.

As a first comparison, 
each CDE method is evaluated according to
estimated CDE loss 
with 95\% confidence intervals.
In this scenario, for all the methods, 
the CDEs were fit using 
the correct number of lags.
For instance, in the 
\textbf{NONLINEAR MEAN} model, 
$3$ lags were used.
\Cref{fig:cde_loss} shows that
\flexts \ significantly outperforms 
GARCH and NNKCDE in all simulations except for 
\textbf{AR} and \textbf{JUMP DIFFUSION}.
In \textbf{JUMP DIFFUSION}, 
\flexts \ and GARCH have similar losses and
outperform NNKCDE.
In \textbf{AR}, GARCH slightly outperforms 
\flexts \ and both of them outperform NNKCDE.
The high performance of GARCH in 
\textbf{AR} is expected, since 
it is the only scenario in which 
the GARCH parametric model is correctly specified.
\begin{figure}
 \centering
 \includegraphics[width=0.8\textwidth]{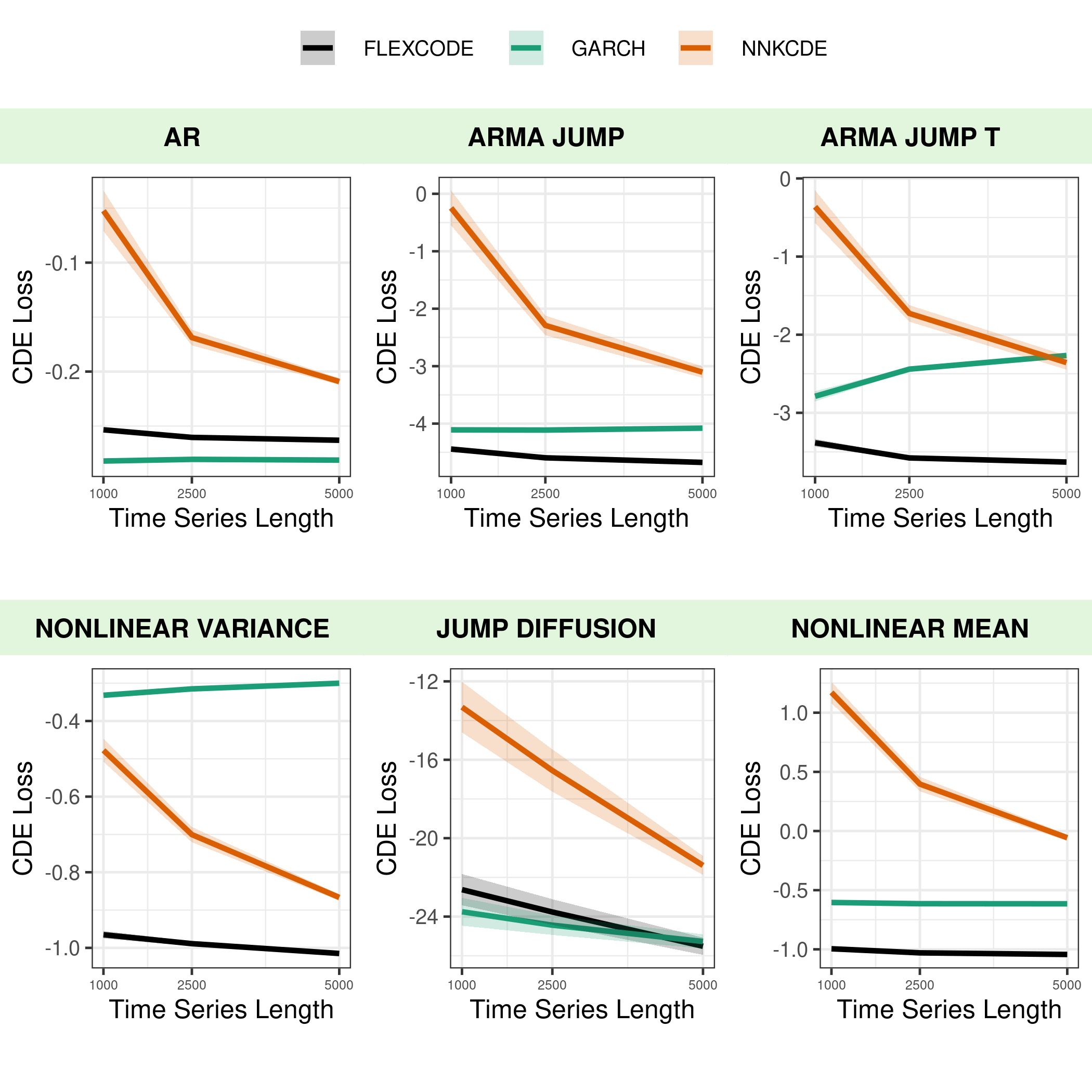}
 \caption{Estimated CDE loss and 
 95\% confidence intervals for 
 \flexts, GARCH and NNKCDE under 
 6 simulation scenarios.
 \flexts \ outperforms GARCH and NNKCDE 
 in all scenarios except for
 \textbf{AR} and \textbf{JUMP DIFFUSION}.} 
 \label{fig:cde_loss}
\end{figure}

The analysis in \cref{fig:cde_loss} can be
complemented by evaluating whether 
the performance of each method is influenced
by the incorrect specification of
the number of lags. 
\Cref{fig::lags} studies this influence in
the \textbf{AR} scenario.
While the performance of GARCH and NNKDCDE is
heavily affected by the chosen number of lags,
the performance of \flexts \ is robust and
does not degrade as the number of lags increases.
As a result, although GARCH outperforms
\flexts \ when the correct 
number of lags (3) is chosen,
\flexts \ as the number of 
adjusted lags increases.

\begin{figure}
 \centering
 \includegraphics[scale=0.3]{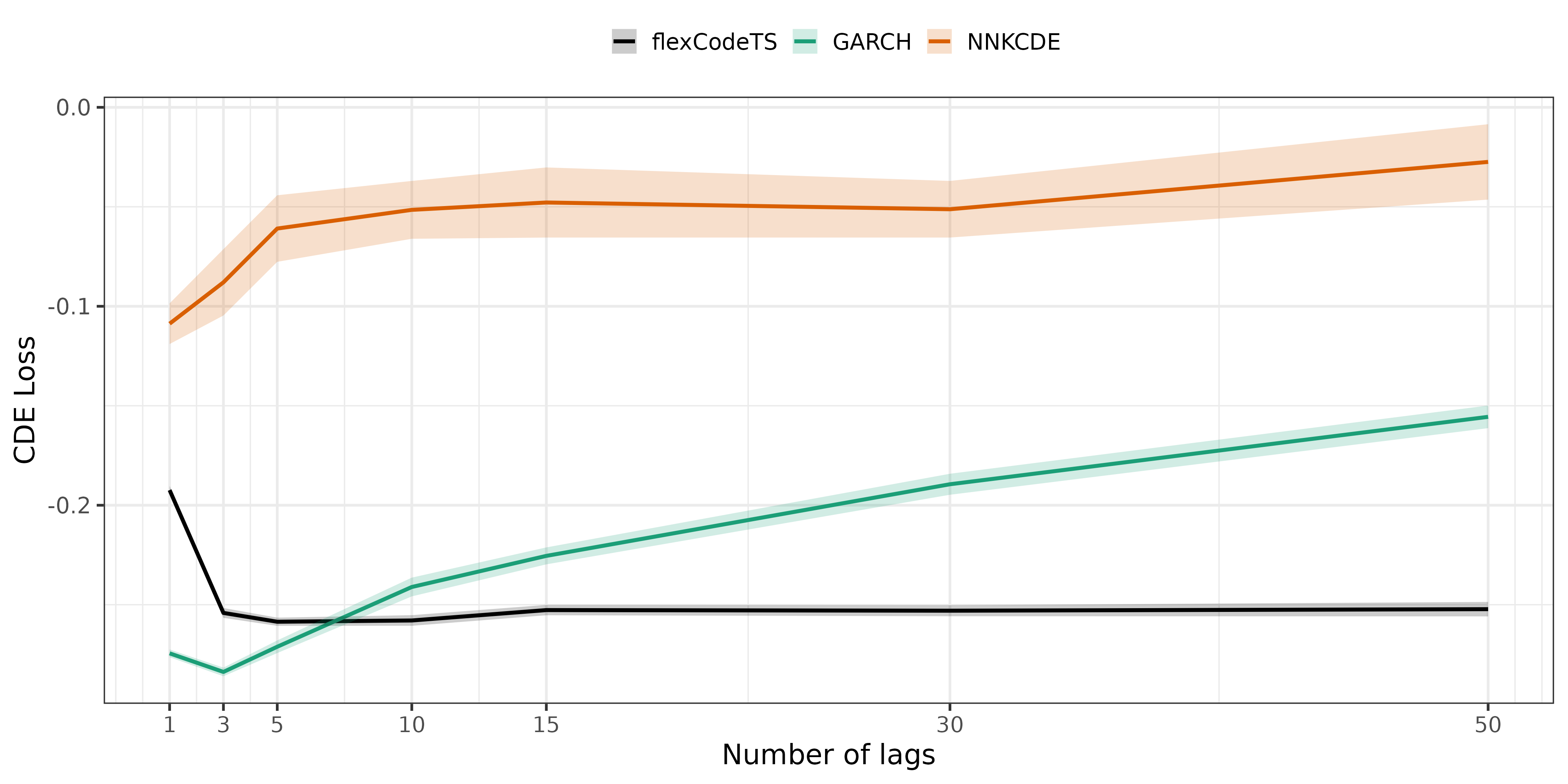}
 \caption{CDE loss for each method as 
 a function of the number of lags that 
 is used. The performance is assessed in
 the \textbf{AR} scenario, where
 the parametric model in GARCH is correct.  
 While GARCH outperforms \flexts \ when 
 the chosen number of lags is correct ($3$),
 \flexts \ outperforms GARCH as
 the number of lags included in the model increases.}
 \label{fig::lags}
\end{figure}

\Cref{fig::importance} explains why
the performance of \flexts \ in 
\textbf{AR} is robust over 
the number of adjusted lags.
We define the importance of each lag using
\flexts \ with Random Forest regression as
average Random Forest importance \citep{breiman2001random} of the lag across 
all regression functions, $\hbu$.
\Cref{fig::importance} shows that,
even when $50$ lags are incorporated
into the model, only 
the first $3$ lags have a high importance.
These lags are precisely the ones that 
are used in \textbf{AR}.
This evidence suggests that \flexts \ 
automatically selects the relevant lags,
which explains why its performance is
robust to the number of adjusted lags.

\begin{figure}
 \centering
 \includegraphics[scale=0.33]{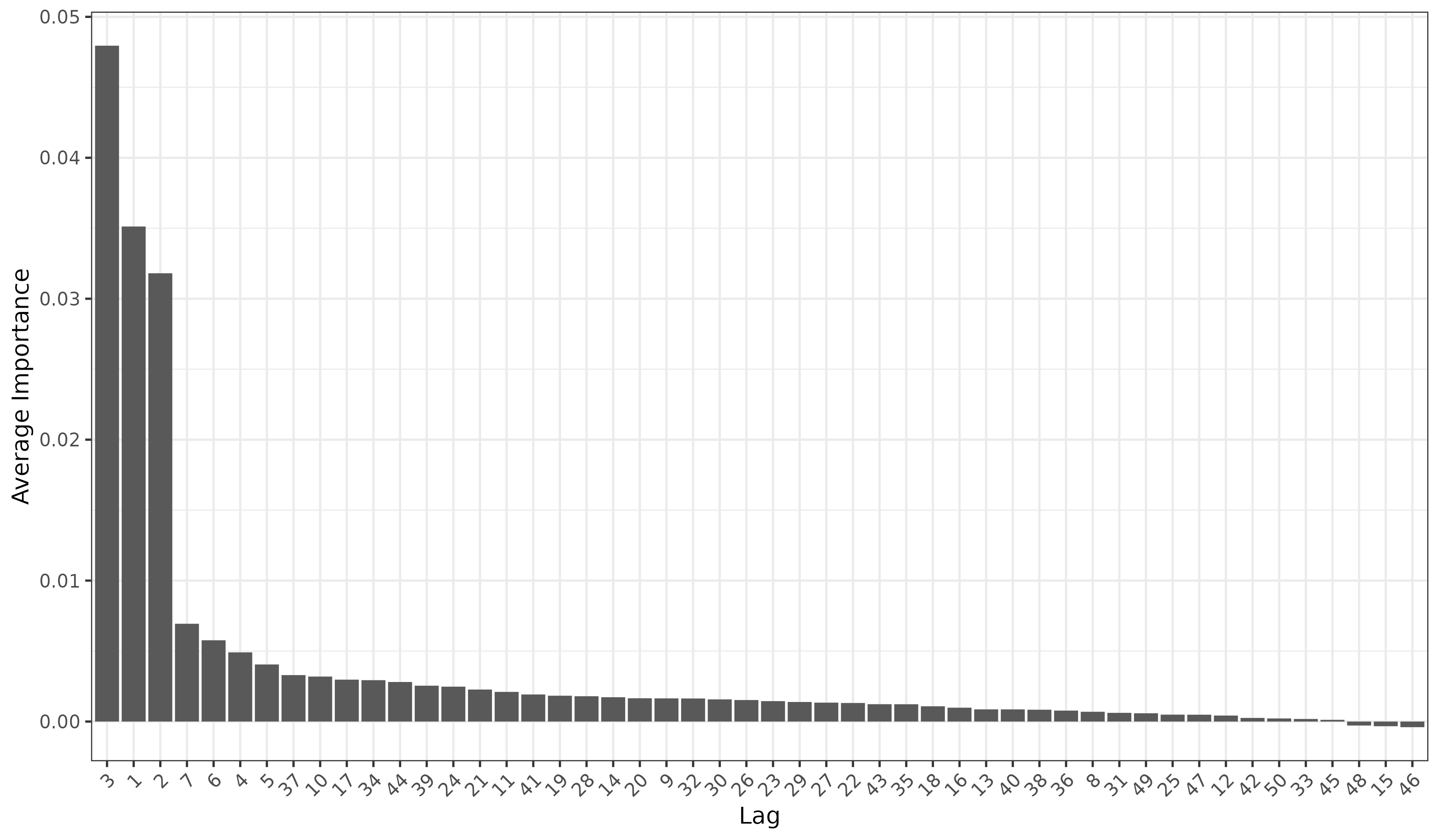}
 \caption{Average importance of each 
 adjusted lag in the \textbf{AR} scenario
 according to \flexts \ with
 Random Forest regression. 
 \flexts \ specifies a high importance to
 the first three lags, which
are the only ones used in this scenario.}
\label{fig::importance}
\end{figure}

\subsection{Applications to real data}
\label{sec:real_data}

This section compares the performance of
\flexts, NNKCDE and GARCH in
applications to real data:
Individual Household Electric Power Consumption from UCI\footnote{Available at \url{http://archive.ics.uci.edu/ml/datasets/Individual+household+electric+power+consumption}.} and
criptocurrency prices.

\subsubsection{Individual Household 
Electric Power Consumption (IHEPC) data}

The IHEPC data contains roughly 
$2 \cdot 10^6$ measurements gathered in 
a house located in France during 47 months.
Although \flexts \ with XGBoost regression 
runs using the full data,
NNKCDE and GARCH experience convergence issues.
In order to be able to compare all methods,
we restrict the analysis to the first
$4 \cdot 10^4$ measurements.
Among these measurements, 
70\% are used for training,
10\% for validation, and
the remaining 20\% for testing the models.
We treat ``global active power'' as 
the response variable and use as predictors 
the ten previous lags and also 
all the other covariates measured at the time,
including dummies for the hour of the day.

\begin{figure}
 \centering
 \includegraphics[scale=0.4]{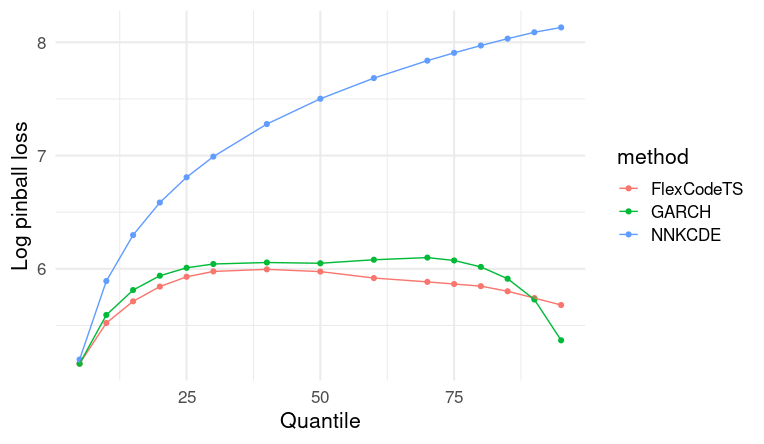}
 \caption{Log pinball loss for GARCH, NNKCDE and
 \flexts \ for each quantile from 5\% to 95\%.
 \flexts \ outperforms NNKCDE and GARCH for
 every quantile except for 95\%.}
 \label{fig::ihepc_pb}
\end{figure}

\flexts \ had a high performance
in the IHEPC data. For instance,
\flexts \ had a CDE loss of -4.4,
lower than those of GARCH (-3.9)
and of NNKCDE (-3.2).
Similarly, \cref{fig::ihepc_pb} presents
each method's log pinball loss for 
each quantile from 5\% to 95\%.
\flexts \ has the lowest pinball loss for
every quantile, except for 95\%.

\Cref{fig::importance_ihepc} shows 
the importance of each variable in \flexts.  
The most relevant variable is 
the global reactive power.
While the global active power measures
the electrical consumption by appliances,
the global reactive power measures
how much electricity bounces back and forth
without any usage.
The hour variables also have a high importance, 
which shows that the predictions
for electrical usage also
depend on the time of the day.

\begin{figure}
 \centering
 \includegraphics[trim=
 {0 0.4cm 0 0},clip,scale=0.35]
 {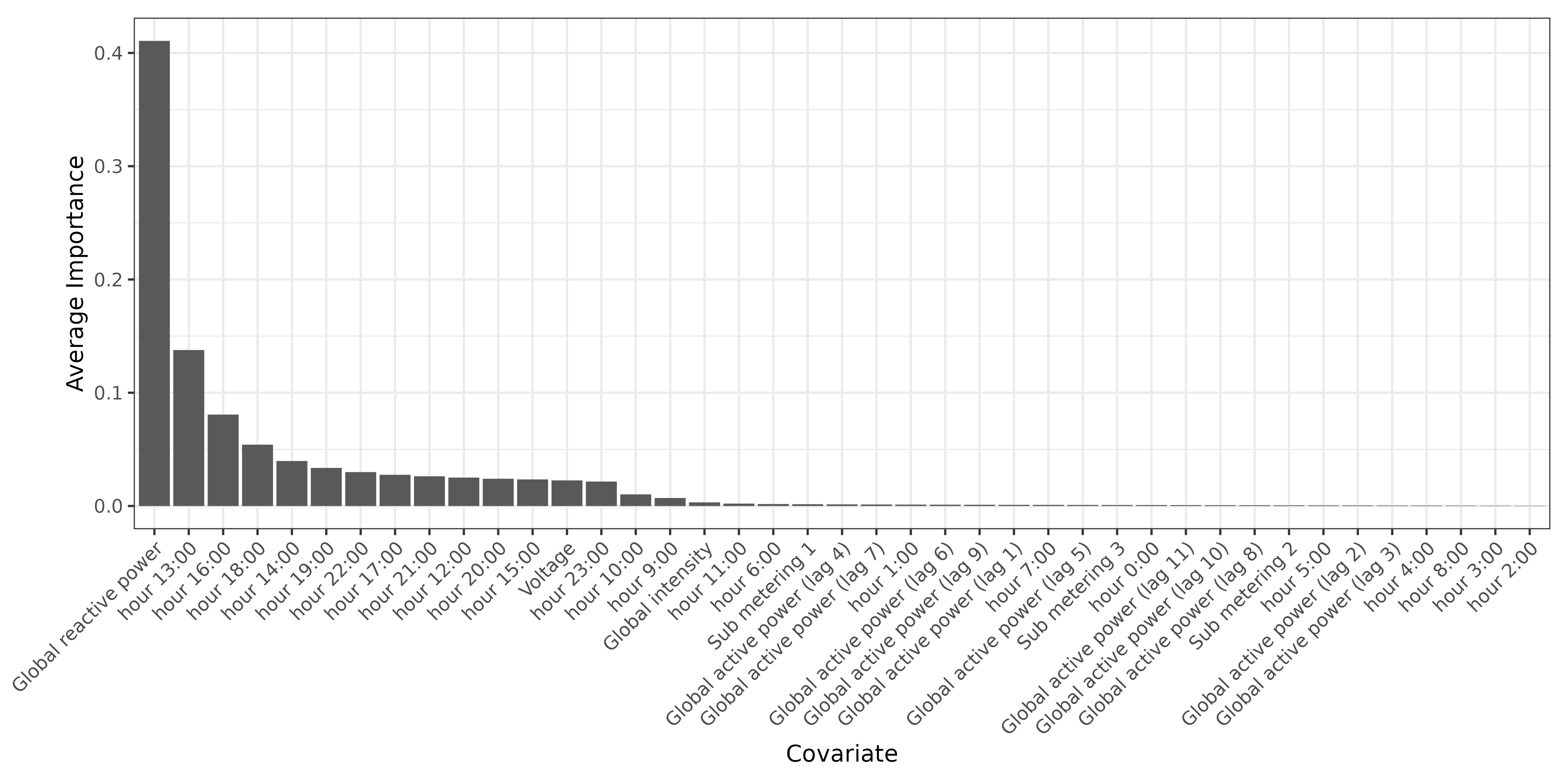}
 \caption{Average importance of 
 predictors according to \flexts \ with 
 XGBoost regression in the IHEPC data.}
 \label{fig::importance_ihepc}
\end{figure}

\subsubsection{Currency data}

This section studies 
conditional density estimation in currency data.
The currency data is composed of
daily measurements from 11/2017 to 
11/2022 of the S\&P500 index,
the exchange rate from euro to dollar, and
closing prices for Ether and Bitcoin.
Due to their high volatility,
cryptocurrency prices provide
a useful complement to the previous analysis.

Using this data, we compare 
the performance of \flexts \ to 
that of GARCH. NNKCDE failed to run with
the available RAM and, thus, was not included.
Data from 2017 to 2020 was used for training,
from 2020 to 2021 for validation, and
from 2021 to 2022 for testing.
In all comparisons, the response variable is
one of the $4$ chosen series.
The following variables were added as predictors:
5 lags of daily return, daily amplitude,
variance of prices in the last $3$ days,
minimum and maximum returns in the last $5$ days.
\flexts \ was trained with LASSO regression.

\begin{figure}
 \centering
 \includegraphics[scale=0.35]
 {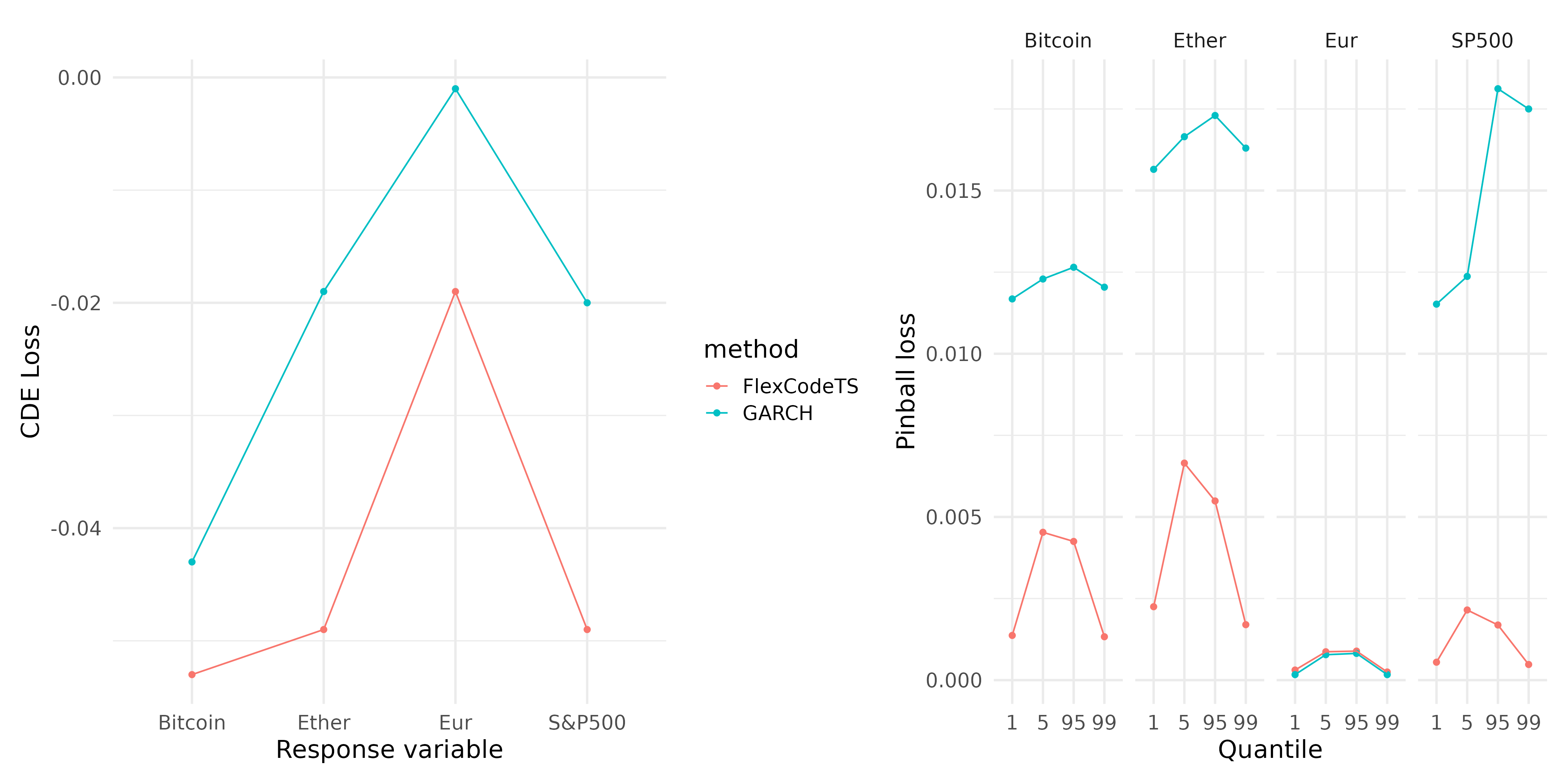}
 \caption{CDE (left) and pinball (right) 
 losses for FlexTS and GARCH when 
 applied to $4$ time series.
 For the CDE loss, FlexTS outperforms 
 GARCH in all of the series.
 For the pinball loss, FlexTS outperforms
 GARCH in all quantiles for the 
 Bitcoin, Ether, and S\&P500 series. 
 For the Eur series, 
 FlexTS and GARCH have 
 similarly small pinball losses 
 in all quantiles.}
 \label{fig::crypto_loss}
\end{figure}

\Cref{fig::crypto_loss} shows that
\flexts \ generally outperforms GARCH 
for all of the response variables.
The graph on the left side presents
the CDE loss for each method.
For all of the $4$ series, 
the CDE loss is higher for
GARCH than for FlexTS.
The graph on the right side compares
the pinball loss of each method for
quantiles $1\%$, $5\%$, $95\%$, and $99\%$.
For the Bitcoin, Ether and S\&P500 series,
the pinball loss for GARCH is
Higher than that for FlexTS in all quantiles.
The only exception to this thread occurs
in the Eur series, in which 
both FlexTS and GARCH have similarly
small pinball losses.

\section{Conclusions}

This paper introduces \flexts, 
a new time series
conditional density estimator.
\flexts \ is a flexible nonparametric CDE,
which can inherit the properties of
arbitrary regression methods. For instance,
in this paper \flexts \ was implemented 
based on penalized linear regression, 
random forests, and XGBoost.
As a future improvement, \flexts \ 
could also be implemented based on
machine learning methods for time series,
such as long short-term memory models
\citep{hochreiter1997long},
attention-based networks 
\citep{vaswani2017attention}, and
gated recurrent units networks
\citep{cho2014properties}.

\flexts \ also inherits 
useful properties from 
the chosen regression method. 
For instance, predictor importance scores
are obtained with
XGBoost and Random Forest.
Moreover, from a theoretical perspective,
we show that \flexts \ inherits
the rate of convergence of 
the chosen regression method.
Hence, \flexts \ can adapt its
convergence by employing 
the regression method that 
best fits the structure of data.

Empirical studies with 
simulated and real data show that
\flexts \ performs well, 
when compared to other CDE methods
for time series, such as
NNKCDE and GARCH.
\flexts \ generally has 
the best performance amongst
the compared methods 
when either the CDE or 
the pinball loss are considered.

%Besides this, FlexCodeTS, contrary to other CDE,  allows n-step ahead predictions to be done without the need of retraining the model. Moreover, Moving average components can be easily incorporated into FlexCodeTS by using time-series regression methods that deal with moving averages, such as ARIMA, SARIMA, VARMA and GARCH models.

\section*{Acknowledgments}

We thank Marcelo Fernandes for
early feedback and discussions on this work.
This work was supported by the Silicon Valley Community Foundation (SVCF; grant \#2018-188547).
Rafael Izbicki is grateful for the financial support of FAPESP (grant 2019/11321-9) and CNPq (grants 309607/2020-5 and 422705/2021-7).
Rafael Stern is grateful for the financial support of FAPESP Research, Innovation and Dissemination Center for Neuromathematics (grant
\#2013/07699-0).

\bibliography{mybibfile}

\appendix 

\section{Proofs}

\subsection{Proof of \cref{thm::main}}

\begin{notation}
 Let $\|g\|_2^2 = \int g(y)^2 dy$ and
 $f_I(y|\u) = \sum_{i=0}^{I} \beta_{i}(\u)\phi_i(y)$. 
\end{notation}

\begin{lemma}[Approximation error] 
 \label{lemma:approximation}
 Under \Cref{ass:stationary,ass:sobolevZ},
 \begin{align*}
  \int \|f(y|\u) - f_I(y|\u)\|_2^2 d\P(\u)
  &= \sum_{i>I} \int \beta_i^2(\u) d\P(\u)
  = O\left(I^{-2\gamma}\right).
 \end{align*}
\end{lemma}

\begin{proof}[Proof of \cref{lemma:approximation}]
 Since $f(\cdot|\u) \in  
 W_{\phi}(s(\u),c(\u))$,
 \begin{align}
  \label{eq:approx_1}
  I^{2s(\u)}\sum_{i \geq I}\beta^2_i(\u) 
  &\leq \sum_{i \geq I} i^{2s(\u)}\beta_i^2(\u) 
  \leq c^2(\u).
 \end{align}
 \begin{align*}
  &\int \|f(y|\u) - f_I(y|\u)\|_2^2 d\P(\u) \\
  =& \int \left\|\sum_{i > I}
  \beta_i(\u)\phi_i(y)\right\|_2^2 d\P(\u) 
  & f(y|\u) = \sum_{i \geq 0}\beta_i(\u)\phi_i(y) \\  
  =& \sum_{i \geq I} \int
  \beta^2_i(\u) d\P(\u)
  & \text{Orthonormality of } \phi \\
  =& \int \sum_{i \geq I} \beta^2_i(\u) d\P(\u) \\
  \leq& \int \frac{c^2(\u)}
  {I^{2s(\u)}} d\P(\u) 
  & \text{\cref{eq:approx_1}} \\
  \leq& 
  \E[c^2(\U)] \cdot I^{-2\gamma} 
  \leq C \cdot I^{-2\gamma} 
  = O\left(I^{-2 \gamma}\right)
  & \text{\cref{ass:sobolevZ}}
 \end{align*}
\end{proof}

\begin{lemma}[Variance Bound] 
 \label{lemma:varianceBound}
 Under \cref{ass:stationary,ass:regression},
 \begin{align*}
  \int \|\hfi - f_I(y|\u)\|_2^2 d\P(\u)
  \leq  I \cdot O_P\left(T^{-2r}\right)
 \end{align*}
\end{lemma}

\begin{proof}[Proof of \cref{lemma:varianceBound}]
 \begin{align*}
  &\int \|\hfi - f_I(y|\u)\|_2^2 d\P(\u) \\
  =& \int \left\| \sum_{i=0}^{I} 
  ((\hbu-\bu)\phi_i(y) \right\|_2^2 d\P(\u) \\
  =&  \sum_{i=0}^I \int \left(
  \hbu - \bu \right)^2d\P(\u)
  & \text{Orthonormality of } \phi \\ 
  &\leq I \cdot \sup_i \int \left(
  \hbu - \bu \right)^2d\P(\u) \\
  &= I \cdot O_P\left(T^{-2r}\right)
  & \text{\cref{ass:regression}}
 \end{align*}
\end{proof}

\begin{proof}[Proof of \cref{thm::main}]
 \begin{align*}
  &\iint\left(\hfi -f(y|\u)\right)^2dyd\P(\u) \\
  =& \int \|\hfi - f(y|\u)\|_2^2 d\P(\u) \\
  =& \int \|(\hfi - f_I(y|\u)) + 
  (f_I(y|\u) - f(y|\u))\|_2^2 d\P(\u) \\ 
  \leq& \int \|\hfi - f_I(y|\u)\|_2^2 d\P(\u) +
  \int \|f_I(y|\u) - f(y|\u)\|_2^2 d\P(\u) \\
  \leq& I \cdot O_P\left(T^{-2r}\right)
  + O\left(I^{-2 \gamma}\right) 
  & \text{\cref{lemma:approximation,lemma:varianceBound}}
 \end{align*}
\end{proof}

\subsection{Proof of \cref{thm:NW}}

This section follows closely 
the proof of Theorem 2 in \cite{Truong1992}.
In order to better compare proofs with
this reference,
the following notation is used:

\begin{definition}
 Let $K_j(\u) = 
 \I(\|\U_j - \u\|_2 \leq \delta_T)$ be
 the uniform kernel. 
 Also, $W_{i,j} 
 = \phi_i(Y_j) - \beta_i(\U_j)$.
\end{definition}

\begin{lemma}
 \label{lemma:truong_1}
 Under \cref{ass:stationary,ass:densityU,ass:bounded_basis},
 $\sup_{i,j,\u} \V[K_j W_{i,j}]
 = O(\delta_T^d)$.
\end{lemma}

\begin{proof}
 \begin{align*}
  \sup_{i,j,\u} \V[K_j W_{i,j}]
  &\leq \sup_{i,j,\u} \V[\E[K_j W_{i,j,\u}|\U_j]] 
  + \sup_{i,j,\u} \E[\V[K_j W_{i,j}|\U_j]] \\
  &= \sup_{i,j,\u} \V[K_j \E[W_{i,j}|\U_j]] 
  + \sup_{i,j,\u} E[K_j^2 \V[ W_{i,j}|\U_j]]
  & K_j \text{ is } \U_j\text{-measurable} \\
  &= \sup_{i,j,\u} E[K_j^2 \V[ W_{i,j}|\U_j]]
  & \E[W_{i,j|\U_j}] \equiv 0 \\
  &\leq M^2 \sup_{j,\u} \E[K_j^2] = O(\delta_T^d)
  & \text{\cref{ass:bounded_basis}}
 \end{align*}
\end{proof}

\begin{lemma}
 \label{lemma:truong_2}
 Under \cref{ass:stationary,ass:bounded_basis,ass:densityU,ass:alphaMixing}
 There exists $c > 0$ such that
 \begin{align*}
  \sup_{i,\u} |Cov[K_{j_1} W_{i,{j_1}}, 
  K_{j_2} W_{i, {j_2}}]| 
  \leq c \delta_T^{\frac{2d}{\nu}}
  \min(\rho^{|j_2-j_1|}, \delta_T^{2d})^{1-2/\nu}
 \end{align*}
\end{lemma}

\begin{proof}
 First observe that, for $\nu > 2$,
 \begin{align}
  \label{eq:truong_2}
  \sup_{i,\u} |\E[K_{j_1} W_{i,{j_1}} 
  K_{j_2} W_{i, {j_2}}]|
  &\leq \sup_{i,\u} \E[K_j |W_{i,j}|^\nu]^{2/\nu} 
  \E[K_{j_1} K_{j_2}]^{1-2/\nu}
  & \text{Holder's inequality} \nonumber \\
  &\leq 4M^2 \delta_T^{\frac{2d}{\nu}}
  \sup_\u\E[K_{j_1} K_{j_2}]^{1-2/\nu}
  & \text{\cref{ass:bounded_basis}} \nonumber \\
  &\leq 4M^2 2^d \delta_T^{\frac{2d}{\nu}} 
  (c_1\delta_T^{2d})^{1-2/\nu}
  & \text{Lemma 3 in \cite{Truong1992}}
 \end{align}
 Next, since $\E[W_{i,j}|\U_j] \equiv 0$ 
 and $K_j$ is $\U_j$-measurable,
 $\E[K_j W_{i,j}] = 0$. Hence,
 \begin{align*}
  &\sup_{i,\u} |Cov[K_{j_1} W_{i,{j_1}}, 
  K_{j_2} W_{i, {j_2}}]| \\
  =& \sup_{i,\u} |\E[K_{j_1} W_{i,{j_1}} 
  K_{j_2} W_{i, {j_2}}]| \\ 
  \leq& 4M^2 2^d \delta_T^{\frac{2d}{\nu}} 
  \min(\alpha(|j_2 - j_1|), 
  c_1\delta_T^{2d})^{1-2/\nu}
  & \text{\cref{eq:truong_2} and 
  Lemma 2 in \cite{Truong1992}} \\
  \leq& 4M^2 2^d \max(8,c_1) 
  \delta_T^{\frac{2d}{\nu}}
  \min(\rho^{|j_2-j_1|}, \delta_T^{2d})^{1-2/\nu}
  & \text{\cref{ass:alphaMixing}}
 \end{align*}
\end{proof}

\begin{lemma}
 \label{lemma:truong_3}
  Under \cref{ass:stationary,ass:bounded_basis,ass:densityU,ass:alphaMixing},
  if  $\delta_T \sim T^{-r}$ for $r > 0$, then
  \begin{align*}
   \sup_{i,\u} \V\left[\sum_{j=1}^T K_j W_{i,j}\right]
   = O(T \delta_T^d)
  \end{align*}
\end{lemma}

\begin{proof}
 \begin{align*}
  \sup_{i,\u} 
  \V\left[\sum_{j=1}^T K_j W_{i,j}\right] 
  &\leq \sum_j \sup_{i,\u} \V[K_j W_{i,j}]
  + \sum_{j_1 \neq j_2} \sup_{i,\u} 
  |Cov[K_{j_1} W_{i,j_1}, K_{j_2} W_{i,j_2}]| \\
  &\leq T \cdot \sup_{i,j,\u} \V[K_j W_{i,j}]
  + T \sum_{j=1}^T 
  \sup_{i,\u} |Cov[K_{0} W_{i,0}, K_{j} W_{i,j}]|
  & \text{\cref{ass:stationary}} \\
  &\leq O(T \delta^d_T)
  + cT \delta_T^{\frac{2d}{\nu}} \cdot 
  \sum_{j=1}^T
  \min(\rho^{j}, \delta_T^{2d})^{1-2/\nu}
  & \text{\cref{lemma:truong_1,lemma:truong_2}} \\
  &= O(T \delta_T^d)
  & \delta_T \sim T^{-r}
 \end{align*}
\end{proof}

\begin{lemma}
 \label{lemma:lip}
 Under \Cref{ass:smoothness},
 $\sup_{i,\u} 
 \frac{\sum_j K_j(\beta_i(\U_j) - \beta_i(\u))}
 {\sum_j K_j} = o(\delta_T)$. 
\end{lemma}
	
\begin{proof}
 Since $\sup_{i,\u} 
 \frac{\sum_j K_j(\beta_i(\U_j) - \beta_i(\u))}
 {\sum_j K_j} \leq 
 \sup_{i,j,u} K_j|\beta_i(\U_j)-\beta_i(\u)|$,
 it is sufficient to show that
 $\sup_{i,j,u} K_j|\beta_i(\U_j)-\beta_i(\u)| 
 = o(\delta_T)$.
 \begin{align*}
   & \sup_{i,j,\u}K_j|\beta_i(\U_j)-\beta_i(\u)| \\
  =& \sup_{i,j,\u} K_j\left|\int \phi_i(y)f(y|\U_j)dy
  -\int \phi_i(y)f(y|\u)dy \right| \\
  \leq& \sup_{i,j,\u}K_j\int \left| \phi_i(y)\right| 
  \left| f(y|\U_j)-f(y|\u) \right|dy \\
  \leq& M_0 \sup_{j,\u} \left(K_j\|\U_j-\u\|\right)
  \sup_i \int \left| \phi_i(y)\right| dy 
  & \text{\Cref{ass:smoothness}} \\ 
  \leq& M_0 \delta_T
  \sqrt{\int 1^2 dy} \cdot 
  \sqrt{\int \phi^2(y)dy} 
  & \text{Cauchy-Schwarz inequality} \\
  =& |\sY|M_0 \delta_T
  & \phi_i \text{ is orthonormal and 
  \Cref{ass:smoothness}}
 \end{align*}	
\end{proof}

\begin{lemma}
 \label{lemma:truong_4}
 Under \cref{ass:stationary,ass:smoothness,ass:bounded_basis,ass:densityU,ass:alphaMixing},
 if $\delta_T \sim T^{-1/(2+d)}$, then
 \begin{align*}
  \sup_i \left(\int \left(\hbu 
  - \bu \right)^2 d\P(\u)\right)^{0.5} 
  &= O_P\left(T^{-1/(2+d)}\right)
 \end{align*}
\end{lemma}

\begin{proof}
 Let $Z = \frac{\sum_{j} K_j \beta_i(\U_j)}
 {\sum_j K_j}$. Using the triangular inequality,
 \begin{align*}
  &\sup_i \left(\int \left(\hbu 
  - \bu \right)^2 d\P(\u)\right)^{0.5} \nonumber \\
  \leq& \sup_i \left(\int \left(\hbu 
  - Z \right)^2 d\P(\u)\right)^{0.5}
  + \sup_i \left(\int \left(Z 
  - \bu \right)^2 d\P(\u)\right)^{0.5}
 \end{align*}
 Hence, it is enough to show that both
 summands are $O_P(T^{-1/(2+d)})$:
 \begin{align*}
  \sup_i \left(\int \left(Z 
  - \bu \right)^2 d\P(\u)\right)^{0.5}
  =& \sup_i \left(\int \left( \frac{K_j(\beta_i(\U_j)-\beta_i(\u))}
  {\sum_j K_j} 
 \right)^2 d\P(\u)\right)^{0.5} \\ 
 \leq& \sup_{i,j,u} K_j|\beta_i(\U_j)-\beta_i(\u)| \\
 =& O_P(\delta_T)
 & \text{\cref{lemma:lip}}
 \end{align*}
 It remains to show that
 $\sup_i \left(\int \left(\hbu - Z \right)^2 d\P(\u)\right)^{0.5} = O_P(T^{-1/(2+d)})$.
 For every $c > 0$,
 
 \begin{align*}
  & \P\left(\sup_i \left( \int \left( 
  \hbu - Z \right)^2 d\P(\u) \right)^{0.5} \geq
  c\left(T^{-1}\delta_T^{-d}\right)^{0.5} \right) \\
  =& \P\left(\sup_i \int \left( 
  \frac{\sum_j{K_j W_{i,j}}}
  {\sum_j K_j} \right)^2 d\P(\u)  \geq
  c^2 T^{-1}\delta_T^{-d} \right) \\
  \leq& \P\left(\sup_i \int \left( 
  \sum_j{K_j W_{i,j}}\right)^2 d\P(\u)
  \geq c^2 c_3^2 T\delta_T^{d} \right)
  & \text{Lemma 7 in \cite{Truong1992}} \\ 
  \leq& \P\left(\sup_{i,\u} \left( 
  \sum_j{K_j W_{i,j}}\right)^2
  \geq c^2 c_3^2 T\delta_T^{d} \right) \\
  \leq& \frac{\sup_{i,\u} \V\left[\sum_j{K_j W_{i,j}}\right]}{c^2 c_3^2 T\delta_T^{d}}
  & \E[K_j W_{i,j}] = 0 \\
  =& \frac{O(T \delta_T^d)}
  {c^2 c_3^2 T\delta_T^{d}}
  & \text{\cref{lemma:truong_3}}
 \end{align*}
\end{proof}

\begin{proof}[Proof of \cref{thm:NW}]
 It follows from \cref{lemma:truong_4} that
 \cref{ass:regression} is satisfied for
 $r = \frac{1}{2+d}$. The proof follows
 directly from applying \cref{thm::main}.
\end{proof}

\end{document}